\documentclass[journal,onecolumn]{IEEEtran}
\usepackage{amsmath}
\usepackage{booktabs}
\usepackage{latexsym}
\usepackage{mathrsfs}
\usepackage{amsthm}
\usepackage{amssymb}
\usepackage{amsfonts}
\usepackage{amsbsy}

\usepackage[shortlabels]{enumitem}
\usepackage{url}
\usepackage{array}
\usepackage{pdflscape}
\usepackage{xcolor}
\usepackage{stmaryrd}
\usepackage{verbatim}
\usepackage{braket}

\newcommand{\Ff}{{\mathbb F}}

\newcommand\rank{\operatorname{rank}}   

\newcommand\cc{{\mathcal C}}        %

\def\Tr{\operatorname{Tr}}

\theoremstyle{plain}
\newtheorem{thm}{Theorem}
\newtheorem{lem}[thm]{Lemma}

\theoremstyle{definition}

\newtheorem{remark}{Remark}
\def\Tr{\operatorname{Tr}}

\usepackage[colorlinks=true]{hyperref}

\begin{document}
\title{Entanglement-Assisted and Subsystem Quantum Codes: New Propagation Rules and Constructions}

\author{\centerline{Gaojun Luo, Martianus Frederic Ezerman, and San Ling}
\thanks{G. Luo, M. F. Ezerman, and S. Ling are with the School of Physical and Mathematical Sciences, Nanyang Technological University, 21 Nanyang Link, Singapore 637371, e-mails: $\{\rm gaojun.luo, fredezerman, lingsan\}$@ntu.edu.sg.}
\thanks{G. Luo, M. F. Ezerman, and S. Ling are supported by Nanyang Technological University Research Grant No. 04INS000047C230GRT01. }

}


\maketitle

\begin{abstract}
This paper proposes new propagation rules on quantum codes in the entanglement-assisted and in quantum subsystem scenarios. The rules lead to new families of such quantum codes whose parameters are demonstrably optimal. To obtain the results, we devise tools to puncture and shorten codes in ways that ensure their Hermitian hulls have certain desirable properties. More specifically, we give a general framework to construct $k$-dimensional generalized Reed-Solomon codes whose Hermitian hulls are $(k-1)$-dimensional maximum distance separable codes.
\end{abstract}


\section{Introduction}\label{sec:intro}

Let $q$ be a prime power and let $\Ff_q$ be the finite field with $q$ elements. The {\it weight} of a vector $\mathbf{v} \in \Ff_q^n$, denoted by ${\rm wt}(\mathbf{v})$, is the number of its nonzero entries. Given a nonempty $\mathcal{S} \subseteq \Ff_q^n$, let ${\rm wt} (\mathcal{S}) = \min\{ {\rm wt}(\mathbf{v}) \colon \mathbf{v} \in \mathcal{S}, \mathbf{v}\neq \mathbf{0}\}$. An $[n,k,d]_{q^2}$ code $\cc$ is a $k$-dimensional subspace of $\Ff_{q^2}^n$. Its minimum distance $d$ is the smallest weight of its nonzero codewords and $\cc$ is \emph{maximum distance separable} (MDS) if $d=n-k+1$. The \emph{Hermitian inner product} of vectors $\mathbf{u}=(u_1,\cdots,u_n)$ and $\mathbf{v}=(v_1,\cdots,v_n)$ of $\Ff_{q^2}^n$ is $\langle\mathbf{u},\mathbf{v}\rangle_{{\rm H}}=\sum_{i=1}^nu_iv_i^q$. The \emph{Hermitian dual} of $\cc$ is
\[
\mathcal{C}^{\perp_{\rm H}} = \left\{\mathbf{u} \in \mathbb{F}_{q^2}^n : \langle\mathbf{u},\mathbf{v}\rangle_{{\rm H}}=0,
\mbox{ for all } \mathbf{v} \in \mathcal{C} \right\}.
\]
The intersection $\cc\cap\mathcal{C}^{\perp_{\rm H}}$ is called the \emph{Hermitian hull} of $\cc$ and is denoted by ${\rm Hull}_{\rm H}(\cc)$. The two extremal cases, namely, ${\rm Hull}_{\rm H}(\cc)=\cc$ and ${\rm Hull}_{\rm H}(\cc)=\{\mathbf{0}\}$ are known respectively as \emph{Hermitian self-orthogonal} and \emph{Hermitian linear complementary dual} (LCD) codes.

A $q$-ary \emph{quantum error-correcting code} (QECC), also known as a {\it qudit code}, is a $K$-dimensional subspace of $(\mathbb{C}^q)^{\otimes n}$. We use the respective terms {\it qubit} and {\it qutrit} codes when $q=2$ and $q=3$. The parameters $[[n,\kappa,\delta]]_q$ of a QECC signifies that the code has dimension $q^{\kappa}$ and can correct quantum error operators affecting up to $\lfloor(\delta-1)/2\rfloor$ arbitrary positions in the quantum ensemble. Formalizing the stabilizer framework in \cite{Gottesman1997}, Calderbank {\it et al}. in \cite{Calderbank1998} proposed a general approach to construct qubit QECCs, by group character theory and finite geometry. This method was subsequently extended to the nonbinary case in \cite{Ketkar2006}, establishing the correspondence between a Hermitian self-orthogonal classical code and a stabilizer QECC.

The stabilizer frameworks requires the ingredient classical linear codes to be Hermitian self-orthogonal. One can relax the orthogonality condition while still being able to perform quantum error control in the frameworks of \emph{entanglement-assisted quantum error-correcting codes} (EAQECCs) \cite{Brun2006} and \emph{subsystem codes} \cite{Poulin2005}.

\subsection{Entanglement-Assisted Quantum Error-Correcting Codes}

Brun, Devetak, and Hsieh introduced EAQECCs in \cite{Brun2006}. The sender and the receiver share pairs of error-free maximally entangled states ahead of time. A $q$-ary EAQECC, denoted by $[[n,\kappa,\delta; c]]_q$, encodes $\kappa$ logical qudits into $n$ physical qudits, with the help of $n-\kappa-c$ ancillas and $c$ pairs of maximally entangled qudits. Such a quantum code can correct up to $\lfloor(\delta-1)/2\rfloor$ quantum errors. An EAQECC is a QECC if the code is designed without entanglement assistance, that is, when $c=0$.

With maximally entangled states as an additional resource, the pools of feasible classical ingredients in the construction of EAQECCs can include classical codes which are not self-orthogonal. A general construction of qubit EAQECCs was provided in \cite{Brun2006} via any binary or quaternary linear codes. This approach was generalized to the qudit case in \cite{Galindo2019}.

\begin{lem}{\rm \cite[Theorem 3]{Galindo2019}}\label{prop:two}
If $\mathcal{C}$ is an $[n,k,d]_{q^2}$ code, then there exists an 
$[[n,\kappa, \delta; c]]_q$ EAQECC $\mathcal{Q}$ with
\begin{align*}
c = k - \dim_{\mathbb{F}_{q^2}}
\left({\rm Hull}_{\rm H}(\cc)\right), \
\kappa = n-2k+c \mbox{, and}\
\delta = {\rm wt}\left(\mathcal{C}^{\perp_{\rm H}} \setminus {\rm Hull}_{\rm H}(\cc)\right).
\end{align*}
\end{lem}

If $\cc\subseteq \cc^{\perp_{\rm H}}$ or $\cc^{\perp_{\rm H}}\subseteq \cc$ in Lemma \ref{prop:two}, we obtain a QECC with parameters $[[n,n-2k, \delta_1]]_q$ or $[[n,2k-n, \delta_2]]_q$, where $\delta_1 = {\rm wt}\left(\mathcal{C}^{\perp_{\rm H}} \setminus \cc\right)$ and $\delta_2 = {\rm wt}\left(\mathcal{C} \setminus \cc^{\perp_{\rm H}}\right)$. The code $\mathcal{Q}$ in Lemma \ref{prop:two} is {\it nondegenerate} or {\it pure} if $\delta = d(\cc^{\perp_{\rm H}})$.

In both classical and quantum setups, having propagation rules to derive new codes from a code at hand is useful. In \cite{Ketkar2006}, Ketkar {\it et al}. demonstrated that the existence of a pure $[[n,\kappa,\delta>1]]_q$ QECC implies the existence of an $[[n-1,\kappa+1,\delta-1]]_q$ QECC via \emph{shortened} classical codes. There is, however, no known propagation rule via \emph{punctured} classical codes since puncturing does not generally preserve self-orthogonality. In EAQECCs, on the other hand, puncturing leads to a propagation rule. Galindo {\it et al}. in \cite{Galindo2019,Galindo2021} took a Hermitian self-orthogonal code $\cc$ and showed that an $[[n,\kappa,\delta]]_q$ QECC constructed from $\cc$ by Lemma \ref{prop:two} gives rise to an $[[n-s,\kappa,\geq\delta; s]]_q$ EAQECC for each $s \in \{1,2,\ldots,\delta-1 \}$. Grassl, Huber, and Winter in \cite{GraHubWin2022} generalized the propagation rule by proving that \emph{any pure} $[[n,\kappa,\delta]]_q$ QECC leads to an $[[n-s,\kappa,\geq\delta; s]]_q$ EAQECC for each $s$ in the same range.

The Singleton-like bound for any $[[n,\kappa, \delta; c]]_q$ EAQECC in \cite[Corollary 9]{GraHubWin2022} reads
\begin{alignat}{5}
  \kappa &\le c+\max\{0,n - 2 \delta + 2\},\label{eq:QMDS_small_distance}\\
  \kappa &\le n-\delta+1,\label{eq:QMDS_trivial}\\
  \kappa
  &\le\frac{(n-\delta+1)(c+2\delta-2-n)}{3\delta-3-n} \mbox{, with }
    \delta-1\ge\frac{n}{2}.\label{eq:QMDS_large_distance}
\end{alignat}
We call EAQECCs that achieve equality in \eqref{eq:QMDS_small_distance} whenever $\delta \le \frac{n}{2}$ and in \eqref{eq:QMDS_large_distance} whenever $\delta>\frac{n}{2}$ {\it quantum maximum distance separable} (QMDS). Given a classical $[n,k,n-k+1]_{q^2}$ MDS code whose Hermitian hull has dimension $\ell$, Lemma \ref{prop:two} produces two EAQECCs with respective parameters
\begin{equation}\label{EAQcon}
[[n,k-\ell,n-k+1;n-k-\ell]]_q\ \mbox{and}\ [[n,n-k-\ell,k+1;k-\ell]]_q.
\end{equation}
In general, only one of the two EAQECCs is QMDS since the code with $\delta > n/2$ cannot meet the bound in \eqref{eq:QMDS_large_distance}. Lemma \ref{prop:two} has been a main construction tool for EAQECCs. 
 
In a concurrent work \cite[Section 3]{Luo2021}, Luo {\it et al}. established the constraint 
\begin{equation}\label{eq:QMDS}
 2\delta \le n+c-\kappa+2
\end{equation}
on the parameters of any $[[n,\kappa, \delta; c]]_q$ EAQECC constructed by Lemma \ref{prop:two}.

An EAQECC which is generated based on Lemma \ref{prop:two} is {\it optimal} if its parameters reach equality in the bound in \eqref{eq:QMDS}. An optimal $[[n,\kappa, \delta; c]]_q$ is QMDS whenever $\delta\le\frac{n}{2}$. Thus, each classical MDS code corresponds to two optimal EAQECCs. By determining the dimensions of the Hermitian or Euclidean hulls of MDS codes, many researchers have proposed various constructions of optimal EAQECCs. The works in \cite{Fan2016,Guenda2017,Lu2018,Lu2018a,Luo2019,Fang2020,Gao2021,Li2019,Tian2020,Chen2021,Chen2021a,Qian2019,Qian2017,Wang2019} are representatives of the literature on the topic.

\subsection{Subsystem codes}

A \emph{quantum subsystem code} is a subspace of $(\mathbb{C}^q)^{\otimes n}$ which can be decomposed into a tensor product $\mathcal{A} \otimes \mathcal{B}$ of a $K$-dimensional vector space $\mathcal{A}$ and an $R$-dimensional vector space $\mathcal{B}$. The space $\mathcal{A}$ is the \emph{subsystem}. It stores information. The space $\mathcal{B}$ is the \emph{gauge subsystem}. Quantum errors in $\mathcal{B}$ can be ignored as $\mathcal{B}$ is utilized only to provide some extra redundancy. The notation $[[n,\kappa,r,\delta]]_q$ for a subsystem code signifies that the code is $q$-ary, with $K=q^{\kappa}$, $R=q^r$, and the ability to correct quantum error operators affecting up to $\lfloor(\delta-1)/2\rfloor$ arbitrary positions in the (information) subspace $\mathcal{A}$. Interested readers can consult the chapter \cite{KP2013} written by Kribs and Poulin for the basic theory. 

Like stabilizer QECCs and EAQECCs, subsystem codes can be constructed via linear classical codes over finite fields. The following general construction of subsystem codes was established in \cite{Aly2006}.

\begin{lem}\label{QSCH}{\rm \cite[Corollary 4]{Aly2006}}
Let $\cc$ be an $[n,k,d]_{q^2}$ code with Hermitian hull ${\rm Hull}_{\rm H}(\cc)$ of dimension $\ell$. If $k+\ell<n$, then there exists a subsystem code $\mathcal{Q}$ with parameters
\[
\left[\left[n,n-k-\ell,k-\ell,{\rm wt}\left({\rm Hull}_{\rm H}(\cc)^{\perp_{\rm H}}\setminus \cc\right)\right]\right]_q.
\]
\end{lem}
Subsystem codes constructed by Lemmas \ref{QSCH} are also known as {\it Clifford subsystem codes}. If the minimum distance of such a code is $\delta = d\left({\rm Hull}_{\rm H}(\cc)^{\perp_{\rm H}}\right)$, then the code is {\it pure} or {\it nondegenerate}. An $[[n,\kappa,0,\delta]]_q$ Clifford subsystem code is a QECC.

Aly and Klappenecker gave two propagation rules on pure subsystem codes via extended and shortened classical codes in \cite{SALAHA.2009}. A pure $[[n,\kappa,r,\delta]]_q$ subsystem code propagates to two pure subsystem codes of respective parameters $[[n+1,\kappa,r,\geq\delta]]_q$ and $[[n-1,\kappa+1,r,\delta-1]]_q$. Given an $[[n,\kappa,r,\delta]]_q$ subsystem code, a trade-off among its parameters was stated in \cite{Aly2006} as the \emph{Singleton-like bound} $\kappa+r\leq n-2\delta+2$. A subsystem code that reaches equality in the bound is called {\it optimal}. In particular, a linear code whose Hermitian hull is MDS corresponds to an optimal subsystem code. Using self-orthogonal classical codes, Aly and Klappenecker identified six families of optimal subsystem codes in \cite{SALAHA.2009}. Qian and Zhang constructed a family of optimal subsystem codes with parameters $[[2^{2m}+1,(2^m-1)^2,4,2^m-1]]_{2^m}$ via cyclic codes that are not Hermitian self-orthogonal in \cite{Qian2013}.

\subsection{Our contributions}

In this paper, we propose \emph{new propagation rules on quantum codes based on punctured and shortened classical codes}. We also present some families of optimal EAQECCs and subsystem codes. The motivations and contributions are summarized in three items. The first two come from the quantum setup whereas the third is from classical coding theory.
\begin{enumerate}
\item The parameters of the Hermitian hull of a linear code form a key ingredient in the constructions of EAQECCs and subsystem codes. Given a linear code $\cc$, we denote by $\cc^S$ and $\cc_S$ the respective codes obtained by puncturing and shortening $\cc$ on an index set $S$ of size $|S|=s$. 
\begin{itemize}
    \item Theorem \ref{Hull} gives a simple condition on the index set $S$ to ensure ${\rm Hull}_{\rm H}(\cc^S)={\rm Hull}_{\rm H}(\cc_S)=({\rm Hull}_{\rm H}(\cc))_S$ and establishes the dimension of the code. The Hermitian hulls of both $\cc^S$ and $\cc_S$ are, therefore, the code ${\rm Hull}_{\rm H}(\cc)$ shortened on a suitable $S$.
    
    \item We derive two new propagation rules on EAQECCs by using Theorem \ref{Hull}. The existence of a pure $[[n,\kappa, \delta; c]]_q$ EAQECC $\mathcal{Q}$ constructed by Lemma \ref{prop:two} and a linear code with $\ell$-dimensional Hermitian hull implies the existence of two EAQECCs for each $ s \in \{1,2,\ldots, \ell\}$ with respective parameters 
    \begin{equation}\label{eq:propruleone}
    [[n-s,\kappa, \geq\delta; c+s]]_q \mbox{ and } 
    [[n-s,\kappa+s, \geq\delta-s; c]]_q.
    \end{equation}
    The first propagation rule also holds when the code $\mathcal{Q}$ is not pure. The resulting two quantum codes reach equality in the bound in \eqref{eq:QMDS} if the quantum code $\mathcal{Q}$ is already optimal with respect to the bound in \eqref{eq:QMDS}. Our first new propagation rule generalizes the rule in \cite{Galindo2019,Galindo2021,GraHubWin2022}, in which the initial code $\mathcal{Q}$ must be a \emph{stabilizer} QECC and the range of the variable $s$ is smaller than that in our rule.
    
    \item In a similar manner, we derive two new propagation rules on quantum subsystem codes. The existence of a pure $[[n,\kappa,r,\delta]]_q$ subsystem code $\widehat{\mathcal{Q}}$ derived by Lemma \ref{QSCH} and a linear code with $\ell$-dimensional Hermitian hull ensures the existence of two subsystem codes for each $ s \in \{1,2,\ldots, \ell\}$ with respective parameters 
    \begin{equation}\label{eq:propruletwo}
    [[n-s,\kappa,r+s,\geq\delta-s]]_q \mbox{ and } 
    [[n-s,\kappa+s,r,\geq\delta-s]]_q.
    \end{equation}
Our first propagation rule for subsystem codes also works when the code $\widehat{\mathcal{Q}}$ is not pure. If $\widehat{\mathcal{Q}}$ is optimal, then the resulting two subsystem codes are also optimal.
\end{itemize}

Constructions of EAQECCs and subsystem codes via linear codes whose hulls have certain properties under the \emph{Euclidean} inner product have been given in \cite[Theorem 4]{Galindo2019} and \cite[Corollary 5]{Aly2006}, respectively. Our propagation rules also work, with proper modifications, when the inner product on the classical code ingredients is Euclidean. For brevity, this paper present the results when the inner product is Hermitian and the hulls are defined accordingly.

\item Linear codes with MDS Hermitian hulls correspond to optimal subsystem codes. In Remark \ref{rem:21} we show how to utilize such linear codes to construct Hermitian LCD codes. To the best of our knowledge, only one such construction is previously known from \cite{Qian2013}. Since then, researchers have paid more attention to finding the dimensions of Hermitian hulls of linear codes. A natural but somehow overlooked question in this topic is 
\begin{center}
How can we explicitly construct codes whose Hermitian hulls are MDS?
\end{center}
Our investigation to answer the question leads to the following results.
\begin{itemize}
    \item We study the Hermitian hulls of generalized Reed-Solomon (GRS) codes and provide a general framework to generate $k$-dimensional GRS codes whose hulls are $(k-1)$-dimensional MDS codes in Theorem \ref{GRSMDSHULL}. 
    \item The framework leads to four families of GRS codes with MDS Hermitian hulls. 
    \item Applying our new propagation rules on these families of GRS codes gives us the sixteen families of optimal EAQECCs in Theorem \ref{thmMDS} and the eight families of optimal subsystem codes in Theorem \ref{subMDS}.
    \item Tables \ref{table1}, \ref{table3}, \ref{table2}, and \ref{table4} confirm that most of the constructed optimal quantum codes have new parameters.
\end{itemize}

\item A linear code is Euclidean LCD if its Euclidean hull contains only the zero codeword $\{\mathbf{0}\}$. Such codes play a critical role in orthogonal direct sum masking \cite{Bringer2014,Carlet2016}. In \cite[Proposition 14]{Carlet2016}, Carlet and Guilley briefly discussed the construction of Euclidean LCD codes by puncturing and shortening. In a landmark paper \cite{Carlet2018}, Carlet {\it et al}. proved that any $q$-ary linear code is equivalent to a Euclidean LCD code if $q>3$ and each $q^2$-ary linear code is equivalent to a Hermitian LCD code if $q>2$. We improve on the results in those three references by studying the hulls of punctured and shortened codes.
\begin{itemize}
    \item In Theorem \ref{Hull} we formally state that, given a linear code over arbitrary finite field, we can produce two Euclidean (or Hermitian) LCD codes. One code comes from puncturing whereas the other is the result of shortening. 
    \item Moreover, we provide an \emph{easy-to-implement method} to construct Euclidean or Hermitian LCD codes from any given linear code. A theoretical determination of the exact minimum distances of punctured or shortened codes remains an open problem. For the time being, we resort to the second-best option of performing computations.
\end{itemize}
\end{enumerate}

After this introduction, we present new propagation rules on quantum codes in Section \ref{sec:PR}. We start by devising tools from classical codes based on the puncturing and shortening techniques in the first subsection. The remaining two subsections state the propagation rules on EAQECCs and for subsystem codes, in that order. Based on some measures of optimality and by deploying the tools and results obtained in Section \ref{sec:PR}, Section \ref{sec:CQC} focuses on the constructions of optimal quantum codes in the entanglement-assisted and subsystem frameworks. The first subsection supplies an extensive treatment on the construction of classical codes whose Hermitian hulls are MDS. Families of optimal EAQECCs and optimal subsystem codes, most of them with previously unknown ranges of parameters, are explicitly listed in the subsequent two subsections. We end with some concluding remarks in Section \ref{sec:conclu}.

\section{New propagation rules on quantum codes}\label{sec:PR}

This section presents new propagation rules on EAQECCs and subsystem codes. We assume, throughout, that the respective initial quantum codes are constructed based on Lemmas \ref{prop:two} and \ref{QSCH}.

\subsection{Tools from Punctured and Shortened Classical Codes}

Puncturing and shortening are essential propagation rules in the construction of new codes from old ones. Given a positive integer $n$, we use the shorthand $[n]$ for the set $\{1,\ldots,n\}$. Let $\cc$ be a (not necessarily linear) code over $\Ff_q$ of length $n$ and let $S$ be a subset of $[n]$. The \emph{punctured code} $\cc^S$ is obtained by deleting the components indexed by the set $S$ in each codeword of $\cc$. Let $\cc(S)$ be the collection of all codewords of $\cc$ with entries $0$ on the positions indexed by the set $S$. The \emph{shortened code} $\cc_S$ is derived by puncturing $\cc(S)$ on $S$. On a chosen set $S$, the length of both $\cc^S$ and $\cc_S$ is $n-|S|$.

Based on the Euclidean dual of a code $\cc$, the dimensions of its punctured and shortened codes are calculated in \cite[Theorem 1.5.7]{huffman2003} under certain conditions. For our purpose, we adjust the calculation to the Hermitian form.

\begin{lem}\label{pas}
Let $\cc$ be an $[n,k,d]_{q^2}$ code. If $S$ is a subset of $[n]$ with $|S|=s$, then the following assertions hold.
\begin{enumerate}
\item $(\cc^{\perp_{\rm H}})_S=(\cc^S)^{\perp_{\rm H}}$ and $(\cc^{\perp_{\rm H}})^S=(\cc_S)^{\perp_{\rm H}}$.
\item If $s<d$, then the respective dimensions of $\cc^S$ and $(\cc^{\perp_{\rm H}})_S$ are $k$ and $n-k-s$.
\item If $s<d$ and $\cc$ is an MDS code, then $\cc^S$ and $(\cc^{\perp_{\rm H}})_S$ are MDS codes.
\end{enumerate}
\end{lem}
\begin{proof}
One can prove each of the first two assertions by a method analogous to the one in the proof of \cite[Theorem 1.5.7]{huffman2003}. The third assertion follows from the fact that the code $\cc$ is MDS if and only if it has a minimum weight codeword in any $d$ coordinates 
\cite[Theorem 4 in Chapter 11]{macwilliams1977}.
\end{proof}

Let $\cc$ be an $[n,k,d]_q$ code with generator matrix $G$ and let $\rho$ be a permutation on the set $[n]$. Let $\mathbf{a} =(a_1, \ldots, a_n) \in \left(\Ff_q^*\right)^n$. The code 
\begin{equation}\label{eq:rhoc}
\rho_{\mathbf{a}}(\cc)=\{(a_1c_{\rho(1)},\ldots,a_nc_{\rho(n)}) :(c_1,\ldots,c_n)\in\cc\}
\end{equation}
is \emph{monomially equivalent} to $\cc$. If $a_1= \ldots = a_n = 1$, then $\rho_{\mathbf{a}}(\cc)$ is \emph{permutation equivalent} to $\cc$ and is often written in abbreviated form as $\rho(\cc)$. An \emph{information set} of $\cc$ is a set of $k$ coordinates such that its $k$ corresponding columns in $G$ are linearly independent. The following lemma, whose proof follows immediately from the definition of Hermitian hulls, determines the hulls of punctured codes.

\begin{lem}\label{lemeq}
Let $\cc$ be a $q^2$-ary linear code of length $n$ and let $\rho$ be a permutation on the set $[n]$. The Hermitian hull of the permutation equivalent code $\rho(\cc)$ of $\cc$ is $\rho({\rm Hull}_{\rm H}(\cc))$.
\end{lem}

For each $\mathbf{a} =(a_1,\ldots,a_{n}) \in \mathbb{F}_{q^2}$, we denote by $\mathbf{a}^q$ the corresponding $(a_1^q,\ldots,a_{n}^q)$. Let $\mathcal{C}$ be an $[n,k,d]_{q^2}$ code with generator matrix $G$. Labelling the rows of $G$ as $\mathbf{a}_1,\ldots,\mathbf{a}_k$, we let $G^{\dagger}$ stand for the $n \times k$ matrix whose columns are $\mathbf{a}_1^q,\ldots,\mathbf{a}_k^q$. We call $G^{\dagger}$ the {\it Hermitian transpose} of $G$. Our first major result establishes that the Hermitian hulls of $\cc^S$ and $\cc_S$ are the shortened code $({\rm Hull}_{\rm H}(\cc))_S$. Analogous statements hold for Euclidean hulls and the proofs for the respective claims can be similarly supplied.

\begin{thm}\label{Hull}
Let $\cc$ be an $[n,k,d]_{q^2}$ code with Hermitian hull ${\rm Hull}_{\rm H}(\cc)$ of dimension $\ell$ and let $S \subseteq [n]$ be such that $|S|=s$. Then the following statements hold.
\begin{enumerate}
\item $({\rm Hull}_{\rm H}(\cc))_S\subseteq{\rm Hull}_{\rm H}(\cc^S)$ and $({\rm Hull}_{\rm H}(\cc))_S\subseteq{\rm Hull}_{\rm H}(\cc_S)$.
\item If $S$ is a subset of an information set of ${\rm Hull}_{\rm H}(\cc)$, then 
\begin{equation}\label{eq:samecodes}
{\rm Hull}_{\rm H}(\cc^S) = {\rm Hull}_{\rm H}(\cc_S) = 
({\rm Hull}_{\rm H}(\cc))_S \mbox{ with } 
{\rm dim}({\rm Hull}_{\rm H}(\cc^S))=\ell-s.
\end{equation}
\end{enumerate}
\end{thm}

\begin{proof}
We prove the statements in their given order or appearance.
\begin{enumerate}
    \item Let $\mathbf{c}$ be a codeword of $({\rm Hull}_{\rm H}(\cc))_S$ whose entries in ${\rm Hull}_{\rm H}(\cc)$ are indexed by the set $[n] \setminus S$. We can recover the corresponding $\widetilde{\mathbf{c}}\in{\rm Hull}_{\rm H}(\cc)$ by inserting $0$s in the coordinate positions indexed by $S$. Since $\widetilde{\mathbf{c}}\in\cc$, we have $\mathbf{c}\in\cc_S$ and $\mathbf{c}\in\cc^S$. Since $\widetilde{\mathbf{c}}\in\cc^{\perp_{\rm H}}$, Lemma \ref{pas} confirms that $\mathbf{c}\in(\cc^{\perp_{\rm H}})_S=(\cc^S)^{\perp_{\rm H}}$. Hence, $\mathbf{c}\in {\rm Hull}_{\rm H}(\cc^S)$, which implies $({\rm Hull}_{\rm H}(\cc))_S\subseteq {\rm Hull}_{\rm H}(\cc^S)$. Since $(\cc^{\perp_{\rm H}})_S\subseteq(\cc^{\perp_{\rm H}})^S$, we have $\mathbf{c}\in(\cc^{\perp_{\rm H}})^S$. By Lemma \ref{pas}, $\mathbf{c}\in(\cc^{\perp_{\rm H}})^S=(\cc_S)^{\perp_{\rm H}}$. Since $\mathbf{c}\in\cc_S$, we have $\mathbf{c}\in {\rm Hull}_{\rm H}(\cc_S)$. Thus, $({\rm Hull}_{\rm H}(\cc))_S\subseteq{\rm Hull}_{\rm H}(\cc_S)$.

    \item By Lemma \ref{lemeq}, we can assume without loss of generality that $\begin{pmatrix} I_{\ell} & A \end{pmatrix}$ is a generator matrix of ${\rm Hull}_{\rm H}(\cc)$. We can then stipulate a $k\times n$ generator matrix $G$ of $\cc$ in the form of
    \begin{equation*}\label{eq:GenG}
    G=\begin{pmatrix}
    I_{\ell} & A \\
    O & B
    \end{pmatrix},
    \end{equation*}
    with $O$ being the appropriate zero matrix. It is immediate to verify that $I_{\ell}+AA^{\dagger}$ and $AB^{\dagger}$ are zero matrices and $BB^{\dagger}$ has full rank. Let $\widetilde{I}_{\ell,S}$ be the matrix constructed by deleting the columns of $I_\ell$ whose indices are in $S$. Then the punctured code $\cc^S$ has generator matrix
    \begin{equation*}
    G^S=\begin{pmatrix}
    \widetilde{I}_{\ell,S} & A \\
    O & B
    \end{pmatrix}.
    \end{equation*}
    Thus, we obtain 
    \begin{equation*}
    G^S(G^S)^{\dagger}=\begin{pmatrix}
    \widetilde{I}_{\ell,S}\widetilde{I}_{\ell,S}^{\dagger}+AA^{\dagger} & AB^{\dagger} \\
    BA^{\dagger} & BB^{\dagger}
    \end{pmatrix}
    =\begin{pmatrix}
    \widetilde{I}_{\ell,S}\widetilde{I}_{\ell,S}^{\dagger}-I_\ell & O \\
    O & BB^{\dagger}
    \end{pmatrix}.
    \end{equation*}
    Since $\rank \left(G^S(G^S)^{\dagger}\right)= k-\ell+s$, the dimension of ${\rm Hull}_{\rm H}(\cc^S)$ is $\ell-s$. Deleting the rows of the matrix $\begin{pmatrix} \widetilde{I}_{\ell,S} & A \end{pmatrix}$ indexed by $S$ results in an $(\ell-s)\times (n-s)$ matrix $\begin{pmatrix} I_{\ell-s} & \widetilde{A} \end{pmatrix}$. By how the matrix $A$ is defined, we infer that $\begin{pmatrix} I_{\ell-s} & \widetilde{A} \end{pmatrix} \, (G^S)^{\dagger}$ is the zero matrix. This implies that the code generated by $\begin{pmatrix} I_{\ell-s} & \widetilde{A} \end{pmatrix}$ is a subcode of $(\cc^S)^{\perp_{\rm H}}$. Hence, ${\rm Hull}_{\rm H}(\cc^S)$ has generator matrix $\begin{pmatrix} I_{\ell-s} & \widetilde{A} \end{pmatrix}$.

    Since it is clear that the shortened code $\cc_S$ has generator matrix
    \begin{equation*}
    G_S=\begin{pmatrix}
    I_{\ell-s} & \widetilde{A} \\
    O & B
    \end{pmatrix},
    \end{equation*}
    an argument similar to the one we have just used can confirm that $\begin{pmatrix} I_{\ell-s} & \widetilde{A} \end{pmatrix}$ indeed generates ${\rm Hull}_{\rm H}(\cc_S)$.

    Finally, we consider the shortened code $({\rm Hull}_{\rm H}(\cc))_S$ and note that ${\rm Hull}_{\rm H}(\cc)$ has generator matrix $\begin{pmatrix} I_{\ell} & A \end{pmatrix}$ with $S\subseteq[\ell]$. Deleting the rows and columns of $\begin{pmatrix} I_{\ell} & A \end{pmatrix}$ which are indexed by $S$ yields an $(\ell-s)\times (n-s)$ matrix $\begin{pmatrix} I_{\ell-s} & \widetilde{A} \end{pmatrix}$ that generates $({\rm Hull}_{\rm H}(\cc))_S$. The desired conclusion follows.
\end{enumerate}
\end{proof}

\begin{remark}
It is not true in general that ${\rm Hull}_{\rm H}(\cc^S) = 
{\rm Hull}_{\rm H}(\cc_S) = ({\rm Hull}_{\rm H}(\cc))_S$ for any $S\subseteq[n]$. The following is a counterexample. Let $\Ff_4 = \{0,1,\omega,\omega^2\}$, with $\omega$ being a root of $1 +x + x^2 \in \Ff_2[x]$. Let $\cc$ be a $[28,10,9]_{4}$ code generated by
\[
\begingroup 
\setlength\arraycolsep{3pt}
\begin{pmatrix}
 1&0&0&0&0&0&0&0&w&0&0&1&0&w^2&w&w^2&w&w^2&0&0&0&w^2&w&w^2&1&w&1&w\\
 0&1&0&0&0&0&0&0&1&0&0&1&w&w&w^2&w&1&1&w^2&w^2&w&0&1&0&1&w&1&1\\
 0&0&1&0&0&0&0&0&w&0&0&w^2&w&w&w&1&0&w&w&w&1&w&1&w&w^2&1&1&0\\
 0&0&0&1&0&0&0&0&w&0&0&w&1&w^2&w&w&0&0&w^2&w&w&0&0&w^2&w&1&w&w^2\\
 0&0&0&0&1&0&0&0&1&0&0&w&w^2&0&1&1&0&1&0&1&1&0&w^2&1&0&w&1&w\\
 0&0&0&0&0&1&0&0&w&0&0&0&1&w^2&0&1&w^2&0&w&w^2&w^2&w&1&w^2&w&w&w^2&w^2\\
 0&0&0&0&0&0&1&0&w&0&0&1&0&w^2&w^2&1&w^2&w^2&w&0&0&w&w&0&w^2&w&0&w\\
 0&0&0&0&0&0&0&1&w&0&0&0&w^2&w&1&w^2&0&0&w^2&w&w&w^2&0&w^2&w&0&0&w^2\\
 0&0&0&0&0&0&0&0&0&1&0&1&w&1&w^2&0&w&1&1&1&0&w^2&w&w&w&w^2&w^2&w\\
 0&0&0&0&0&0&0&0&0&0&1&0&w&w^2&w&w&w&w^2&0&w^2&w&w^2&1&1&1&w&w^2&w
\end{pmatrix}.
\endgroup
\]
The codes ${\rm Hull}_{\rm H}(\cc)$ and ${\rm Hull}_{\rm H}(\cc)^{\perp_{\rm H}}$ have respective parameters $[28,1,20]_4$ and $[28,27,1]_4$. If $S=[6]$, then the dimensions of $({\rm Hull}_{\rm H}(\cc))_S$, ${\rm Hull}_{\rm H}(\cc_S)$, and ${\rm Hull}_{\rm H}(\cc^S)$ are $0$, $1$, and $2$, respectively.
\end{remark}


\begin{remark}\label{rem:21}
The conclusion stated as \eqref{eq:samecodes} in Statement 2) of Theorem \ref{Hull} was established based on an information set of the hull of a code $\cc$. The same conclusion holds when we assume that ${\rm Hull}_{\rm H}(\cc)^{\perp_{\rm H}}$ is a code with parameters $[n,n-\ell,d_2]_{q^2}$ and $S$ is an arbitrary subset of $[n]$ with $s<d_2$. The crucial move is to find the minimum distance of ${\rm Hull}_{\rm H}(\cc)$ of an $[n,k,d]_{q^2}$ code $\cc$. If ${\rm Hull}_{\rm H}(\cc)^{\perp_{\rm H}}$ has parameters $[n,n-\ell,d_2]_{q^2}$, then either puncturing or shortening on any $s < d_2$ coordinates of $\cc$ gives a new code whose Hermitian hull is of dimension $\ell-s$. For $\ell - s$ to traverse a large range of values, {\it e.g.}, from $\ell-1$ to $0$, we want $d_2$ to be as close to $\ell+1$ as possible. If $d_2=\ell+1$, then ${\rm Hull}_{\rm H}(\cc)$ is MDS. Thus, we aim to construct a code $\cc$ whose Hermitian hull is MDS. Once we have such a code, the dimensions of the Hermitian hulls of its punctured and shortened codes cover all values less than the the dimension of ${\rm Hull}_{\rm H}(\cc)$.
\end{remark}

The following corollary presents the parameters of punctured and shortened codes constructed based on Theorem \ref{Hull}.

\begin{lem}\label{pasdim}
Keeping the notation in Theorem \ref{Hull}, we derive the following results.
\begin{enumerate}
\item If $S$ is a subset of an information set of ${\rm Hull}_{\rm H}(\cc)$, then $\cc^S$ and $\cc_S$ have respective parameters $[n-s,k,\geq d-s]_{q^2}$ and $[n-s,k-s,\geq d]_{q^2}$.
\item If $S$ is an arbitrary subset of $[n]$ such that $ |S|=s<d_2$, then $\cc_S$ is an $[n-s,k-s,\geq d]_{q^2}$ code. If $d_2\leq d$, then $\cc^S$ is an $[n-s,k,\geq d-s]_{q^2}$ code.
\end{enumerate}
\end{lem}

\begin{proof}
We rely on Lemma \ref{pas} and Theorem \ref{Hull}.
\begin{enumerate}
    \item The lengths and minimum distances of $\cc^S$ and $\cc_S$ follow by definition. Since $S$ is a subset of an information set of $\cc$, the dimension of $\cc_S$ is $k-s$. Hence, $(\cc_S)^{\perp_{\rm H}}$ has length $n-s$ and dimension $n-k$. By Lemma \ref{pas}, we have $(\cc_S)^{\perp_{\rm H}}=(\cc^{\perp_{\rm H}})^S$, which implies that the dimension of $(\cc^{\perp_{\rm H}})^S$ is $n-k$. We know that, if $S$ is a subset of an information set of $\cc$, then the respective dimensions of $\cc_S$ and $(\cc^{\perp_{\rm H}})^S$ are $k-s$ and $n-k$. Since $S$ is a subset of an information set of $\cc^{\perp_{\rm H}}$, replacing $\cc$ by $\cc^{\perp_{\rm H}}$ confirms that the dimension of $\cc^S$ is $k$.

    \item The desired result follows directly from assertion $2$) in Lemma \ref{pas}.
\end{enumerate}
\end{proof}

\subsection{Propagation Rules on EAQECCs}

Given an EAQECC constructed by Lemmas \ref{prop:two}, we can now present two general propagation rules based on the result on the Hermitian hulls of punctured and shortened codes in Theorem \ref{Hull}.

\begin{thm}\label{EAQP}
Let $\mathcal{Q}$ be an $[[n,\kappa, \delta; c]]_q$ EAQECC constructed by a corresponding linear code $\cc$ with $\ell$-dimensional Hermitian hull from Lemma \ref{prop:two}. The code $\mathcal{Q}$ implies the existence of an $[[n-s,\kappa, \geq\delta; c+s]]_q$ EAQECCs for each $s\in[\ell]$.
\end{thm}

\begin{proof}
By Lemma \ref{prop:two}, there exists an $[n,k,d]_{q^2}$ linear code $\cc$ with $\dim_{\Ff_{q^2}}({\rm Hull}_{\rm H}(\cc))=\ell$. Let $S$ be a subset of an information set of ${\rm Hull}_{\rm H}(\cc)$ with $|S|=s$. By Theorem \ref{Hull} and Lemma \ref{pasdim}, we know that $\cc^S$ is an $[n-s,k]_{q^2}$ code with $\dim_{\Ff_{q^2}}({\rm Hull}_{\rm H}(\cc^S))=\ell-s$. It follows from Lemma \ref{pas} and Theorem \ref{Hull} that 
\[
{\rm wt}\left((\cc^S)^{\perp_{\rm H}}\setminus{\rm Hull}_{\rm H}(\cc^S)\right)={\rm wt}\left((\cc^{\perp_{\rm H}})_S\setminus({\rm Hull}_{\rm H}(\cc))_S\right).
\]

We claim that $(\cc^{\perp_{\rm H}})_S\setminus({\rm Hull}_{\rm H}(\cc))_S=\left(\mathcal{C}^{\perp_{\rm H}} \setminus {\rm Hull}_{\rm H}(\cc)\right)_S$. Let $\mathbf{c}$ be a codeword in $(\cc^{\perp_{\rm H}})_S\setminus({\rm Hull}_{\rm H}(\cc))_S$ whose entries are indexed by the set $[n]\setminus S$. We recover $\widetilde{\mathbf{c}}$ of length $n$ from the codeword $\mathbf{c}$ by inserting $0$s in the coordinate positions indexed by $S$. It is evident that $\widetilde{\mathbf{c}}\in\mathcal{C}^{\perp_{\rm H}}$. If $\widetilde{\mathbf{c}}\in{\rm Hull}_{\rm H}(\cc)$, then $\mathbf{c}\in({\rm Hull}_{\rm H}(\cc))_S$, which is a contradiction. Hence, $\widetilde{\mathbf{c}}\in\mathcal{C}^{\perp_{\rm H}} 
\setminus {\rm Hull}_{\rm H}(\cc)$. By how $\widetilde{\mathbf{c}}$ is defined, $\mathbf{c}\in\left(\mathcal{C}^{\perp_{\rm H}} \setminus {\rm Hull}_{\rm H}(\cc)\right)_S$, which implies 
\[
(\cc^{\perp_{\rm H}})_S\setminus({\rm Hull}_{\rm H}(\cc))_S \subseteq\left(\mathcal{C}^{\perp_{\rm H}} \setminus {\rm Hull}_{\rm H}(\cc)\right)_S.
\]
Conversely, the inclusion $\left(\mathcal{C}^{\perp_{\rm H}} \setminus {\rm Hull}_{\rm H}(\cc)\right)_S \subseteq (\cc^{\perp_{\rm H}})_S \setminus({\rm Hull}_{\rm H}(\cc))_S$ can be shown in a similar manner.

One can quickly confirm that 
\[
{\rm wt}\left((\cc^S)^{\perp_{\rm H}}\setminus{\rm Hull}_{\rm H}(\cc^S)\right)={\rm wt}\left(\left(\mathcal{C}^{\perp_{\rm H}} \setminus {\rm Hull}_{\rm H}(\cc)\right)_S\right)\geq\delta.
\]
By Lemma \ref{prop:two}, the code $\cc^S$ leads to an EAQECC with parameters $[[n-s,\kappa, \geq\delta; c+s]]_q$. 
\end{proof}

\begin{remark}
The propagation rule in \cite{Galindo2019,Galindo2021} for EAQECCs says that an $[[n,\kappa,\delta]]_q$ QECC constructed from $\cc$ by Lemma \ref{prop:two} gives rise to an $[[n-s,\kappa,\geq\delta; s]]_q$ EAQECC for each $s\in[\delta-1]$. Grassl, Huber, and Winter in \cite{GraHubWin2022} proved that a pure $[[n,\kappa,\delta]]_q$ QECC implies the existence of an $[[n-s,\kappa,\geq\delta; s]]_q$ EAQECC for each $s\in[\delta-1]$.
Theorem \ref{EAQP} generalizes the collective results of \cite{Galindo2019,Galindo2021,GraHubWin2022} in two directions. First, the initial code is no longer a QECC. Second, the variable $s$ covers a wider range, that is, the number of maximally entangled states traverses all possible values $> c$.
\end{remark}

\begin{thm}\label{EAQP1}
Let $\mathcal{Q}$ be a pure $[[n,\kappa, \delta; c]]_q$ EAQECC constructed via a corresponding linear code $\cc$ with $\ell$-dimensional Hermitian hull by Lemma \ref{prop:two}. The code $\mathcal{Q}$ implies the existence of an EAQECC $\widetilde{\mathcal{Q}}$ with parameters $[[n-s,\kappa+s, \geq\delta-s; c]]_q$ for each $s\in[\ell]$.
\end{thm}
\begin{proof}
Theorem \ref{Hull} and Lemma \ref{pasdim} allow us to deduce that $\cc_S$ is an $[n-s,k-s]_{q^2}$ code with $\dim_{\Ff_{q^2}}({\rm Hull}_{\rm H}(\cc_S))=\ell-s$. Since $\mathcal{Q}$ is pure, ${\rm wt}\left(\cc^{\perp_{\rm H}}\right) = \delta$. By Lemma \ref{pas} and Theorem \ref{Hull}, we arrive at 
\[
{\rm wt}\left((\cc_S)^{\perp_{\rm H}} \setminus{\rm Hull}_{\rm H}(\cc_S)\right) = {\rm wt}\left((\cc^{\perp_{\rm H}})^S \setminus({\rm Hull}_{\rm H}(\cc))_S\right)\geq\delta-s.
\]
Lemma \ref{prop:two} produces an $[[n-s,\kappa+s, \geq\delta-s; c]]_q$ EAQECC $\widetilde{\mathcal{Q}}$ from $\cc_S$.
\end{proof}

\begin{remark}\label{Remark21}
The existence of $\widetilde{\mathcal{Q}}$ in Theorem \ref{EAQP1} is not guaranteed unless $\mathcal{Q}$ is pure. If $(\cc^{\perp_{\rm H}})^S \setminus({\rm Hull}_{\rm H}(\cc))^S = \left(\mathcal{C}^{\perp_{\rm H}} \setminus {\rm Hull}_{\rm H}(\cc)\right)^S$, then, by Theorem \ref{Hull}, we know that $\begin{pmatrix} I_{\ell} & A \end{pmatrix}$ generates ${\rm Hull}_{\rm H}(\cc)$. We can then stipulate an $(n-k)\times n$ generator matrix $G=\begin{pmatrix}
I_{\ell} & A \\
O & B
\end{pmatrix}$ of $\cc^{\perp_{\rm H}}$. Let $\widetilde{I}_{\ell,S}$ be the matrix constructed by deleting the columns of $I_\ell$ whose indices are in $S$. Then $({\rm Hull}_{\rm H}(\cc))^S$ and $(\cc^{\perp_{\rm H}})^S$  have respective generator matrices
$\begin{pmatrix} \widetilde{I}_{\ell,S} & A \end{pmatrix}$ and 
$G^S=\begin{pmatrix} \widetilde{I}_{\ell,S} & A \\ 
O & B \end{pmatrix}$. Hence, one can write 
\begin{multline}
(\cc^{\perp_{\rm H}})^S \setminus({\rm Hull}_{\rm H}(\cc))^S =\left(\mathcal{C}^{\perp_{\rm H}} \setminus {\rm Hull}_{\rm H}(\cc)\right)^S \\
= \{(u_1,\ldots,u_{n-k}) \,G^S : u_i\in\Ff_{q^2} \mbox{ for } i \in[n-k] \mbox{ and there exists } k \mbox{ with } \ell < k \leq n-k \mbox{ such that } u_k \neq 0\}.    
\end{multline}
Since $({\rm Hull}_{\rm H}(\cc))_S\subseteq ({\rm Hull}_{\rm H}(\cc))^S$, we have $\left(\mathcal{C}^{\perp_{\rm H}} \setminus {\rm Hull}_{\rm H}(\cc)\right)^S \subseteq(\cc^{\perp_{\rm H}})^S \setminus({\rm Hull}_{\rm H}(\cc))_S$. In the proof of Theorem \ref{EAQP1}, we are unable to guarantee that ${\rm wt}\left((\cc^{\perp_{\rm H}})^S \setminus({\rm Hull}_{\rm H}(\cc))_S\right) \geq \delta-s$ if there exists a codeword whose weight is less than $\delta-s$ in the gap between $\left(\mathcal{C}^{\perp_{\rm H}} \setminus {\rm Hull}_{\rm H}(\cc)\right)^S$ and $(\cc^{\perp_{\rm H}})^S \setminus({\rm Hull}_{\rm H}(\cc))_S$.
\end{remark}

Combining Theorems \ref{EAQP} and \ref{EAQP1} and the bound in \eqref{eq:QMDS} leads to the conclusion that each \emph{optimal} EAQECC produces \emph{optimal} EAQECC whose range of parameters is listed in the next theorem.

\begin{thm}\label{EAQMDSP}
Let $\mathcal{Q}$ be an $[[n,\kappa, \delta; c]]_q$ optimal EAQECC with respect to the bound in \eqref{eq:QMDS}, with $\mathcal{Q}$ being obtained by Lemma \ref{prop:two} via a linear code $\cc$. Let $\ell$ be the dimension of the Hermitian hull of $\cc$. For each $s\in[\ell]$, the code $\mathcal{Q}$ gives rise to EAQECCs whose respective parameters $[[n-s,\kappa, \delta; c+s]]_q$ and $[[n-s,\kappa+s,\delta-s; c]]_q$ reach equality in the bound in \eqref{eq:QMDS}.
\end{thm}
\begin{proof}
The conclusion follows from Theorems \ref{EAQP} and \ref{EAQP1} and the fact that each  optimal EAQECC is pure.
\end{proof}

\subsection{Propagation Rules on Subsystem Codes}

Aly and Klappenecker demonstrated in \cite{SALAHA.2009} that the existence of a \emph{pure} $[[n,\kappa,r,\delta]]_q$ subsystem code implies the existence of a \emph{pure} $[[n-1,\kappa+1,r,\delta-1]]_q$ subsystem code. Theorem \ref{Hull} characterizes their rule in detail and provides a new propagation rule on quantum subsystem codes based on punctured codes. This new rule does \emph{not} require the initial codes to be pure.

\begin{thm}\label{QSCP}
Let $\mathcal{Q}$ be an $[[n,\kappa,r,\delta]]_q$ subsystem code constructed by Lemma \ref{QSCH} via a corresponding linear code $\cc$. If the dimension of the Hermitian hull of $\cc$ is $\ell$, then the following statements hold.
\begin{enumerate}
\item For each $s\in[\ell]$, the code $\mathcal{Q}$ gives rise to a subsystem code with parameters $[[n-s,\kappa,r+s,\geq\delta-s]]_q$. 
\item If the code $\mathcal{Q}$ is pure, then there exists an $[[n-s,\kappa+s,r,\geq\delta-s]]_q$ subsystem code for each $s \in [\ell]$.
\end{enumerate}
\end{thm}
\begin{proof}
Let $S$ be a subset of an information set of ${\rm Hull}_{\rm H}(\cc)$ such that $|S|=s$. By Theorem \ref{Hull}, there exist a punctured code $\cc^S$ and a shortened code $\cc_S$ of $\cc$ such that ${\rm Hull}_{\rm H}(\cc^S)={\rm Hull}_{\rm H}(\cc_S)=({\rm Hull}_{\rm H}(\cc))_S$. The respective dimensions of ${\rm Hull}_{\rm H}(\cc)$ and $({\rm Hull}_{\rm H}(\cc))_S$ are $\ell$ and $\ell-s$. By Lemma \ref{pasdim}, the respective dimensions of $\cc^S$ and $\cc_S$ are $k$ and $k-s$. By Lemma \ref{pas}, we have $(({\rm Hull}_{\rm H}(\cc))_S)^{\perp_{\rm H}}=({\rm Hull}_{\rm H}(\cc)^{\perp_{\rm H}})^S$. 
Next, we determine the minimum distance of the resulted subsystem codes in the following two cases.
\begin{enumerate}
\item Since $\mathcal{Q}$ has parameters $[[n,\kappa,r,\delta]]_q$, Lemma \ref{QSCH} gives $\delta ={\rm wt}\left({\rm Hull}_{\rm H}(\cc)^{\perp_{\rm H}} \setminus \cc\right)$. Using the same argument as in Remark \ref{Remark21}, we get $({\rm Hull}_{\rm H}(\cc)^{\perp_{\rm H}})^S \setminus \cc^S = \left({\rm Hull}_{\rm H}(\cc)^{\perp_{\rm H}} \setminus \cc\right)^S$. Thus, we arrive at 
\begin{equation}
{\rm wt}\left({\rm Hull}_{\rm H}(\cc^S)^{\perp_{\rm H}} \setminus \cc^S\right) = {\rm wt}\left(({\rm Hull}_{\rm H}(\cc)^{\perp_{\rm H}})^S \setminus \cc^S\right) = {\rm wt}\left(\left({\rm Hull}_{\rm H}(\cc)^{\perp_{\rm H}} \setminus \cc\right)^S\right) \geq \delta-s.
\end{equation}
By Lemma \ref{QSCH}, we obtain an $[[n-s,\kappa,r+s,\geq\delta-s]]_q$ subsystem code from $\cc^S$.
\item If the code $\mathcal{Q}$, whose minimum distance is $\delta$, is pure, then $\delta = {\rm wt}\left({\rm Hull}_{\rm H}(\cc)^{\perp_{\rm H}}\right)$. It is then immediate to confirm that
\begin{equation}
{\rm wt}\left({\rm Hull}_{\rm H}(\cc_S)^{\perp_{\rm H}} \setminus \cc_S\right) = {\rm wt}\left(({\rm Hull}_{\rm H}(\cc)^{\perp_{\rm H}})^S \setminus \cc_S\right) \geq \delta-s.
\end{equation}
Finally, Lemma \ref{QSCH} ensures the existence of an $[[n-s,\kappa+s,r,\geq\delta-s]]_q$ subsystem code from $\cc_S$.
\end{enumerate}
\end{proof}

Theorem \ref{QSCP} shows how an \emph{optimal} subsystem code constructed by Lemma \ref{QSCH} implies the existence of two \emph{optimal} subsystem codes for each $s$ in the stipulated range.

\begin{thm}\label{QSCPMDS}
An $[[n,\kappa,r,\delta]]_q$ optimal subsystem code obtained by Lemma \ref{QSCH} gives rise to two optimal subsystem codes with respective parameters $[[n-s,\kappa,r+s,\delta-s]]_q$ and $[[n-s,\kappa+s,r,\delta-s]]_q$ for each $s\in[\delta-1]$.
\end{thm}
\begin{proof}
One can readily check that an optimal subsystem code constructed by Lemma \ref{QSCH} must be pure. Let $\cc$ be a linear code that corresponds to the optimal $[[n,\kappa,r,\delta]]_q$ subsystem code $\mathcal{Q}$ in Lemma \ref{QSCH}. Let ${\rm Hull}_{\rm H}(\cc)$ be an MDS code of dimension $\delta-1$ and let $S$ be a subset of an information set of ${\rm Hull}_{\rm H}(\cc)$. By \cite[Exercise 317]{huffman2003}, the shortened code $({\rm Hull}_{\rm H}(\cc))_S$ is MDS if $|S|<\delta-1$. Theorem \ref{Hull} guarantees that ${\rm Hull}_{\rm H}(\cc^S)$ and ${\rm Hull}_{\rm H}(\cc_S)$ are MDS. If $|S|=\delta-1$, then ${\rm Hull}_{\rm H}(\cc)^{\perp_{\rm H}}=\Ff_{q^2}^n$. In both cases, Theorem \ref{QSCP} leads directly to the desired conclusion.
\end{proof}

\section{Constructions of optimal quantum codes}\label{sec:CQC}

This section discusses a general framework to produce generalized Reed-Solomon (GRS) codes whose Hermitian hulls are MDS. With the help of these codes, we construct numerous families of optimal EAQECCs and subsystem codes.

\subsection{GRS Codes whose Hermitian Hulls are MDS}

We have just demonstrated that classical GRS codes whose Hermitian hulls are MDS correspond to optimal EAQECCs and optimal subsystem codes. Remark \ref{rem:21} has explained how one can vary the dimensions of the Hermitian hulls. We begin with some known facts regarding GRS codes and cyclic codes.

Let $\Ff_q^*$ denote the multiplicative group of $\Ff_q$. Let $b_1,\ldots,b_n$ be distinct elements of $\Ff_q$ and let $\mathbf{b}=(b_1,\ldots,b_n)$ be a vector of length $n$. Writing $\mathbf{a} = (a_1,\ldots,a_n) \in  (\Ff_q^*)^n$, for each $k\in \{0,1,\ldots,n\}$, the generalized Reed-Solomon (GRS) code $GRS_k(\mathbf{b},\mathbf{a})$ is the evaluation code
\begin{equation}\label{eq:GRS}
GRS_k(\mathbf{b},\mathbf{a}) := \left\{(a_1f(b_1),\ldots,a_nf(b_n)) : f(x) \in\Ff_q[x] \mbox{, with } \deg(f(x)) < k \right\}.
\end{equation}
We know, {\it e.g.}, from \cite[Section 9]{Ling2004}, that $GRS_k(\mathbf{b},\mathbf{a})$ is an $[n,k,n-k+1]_q$ MDS code with generator matrix
\[
G_k=\begin{pmatrix}
a_1 & a_2 &\cdots & a_n\\
a_1b_1 & a_2b_2 &\cdots & a_nb_n\\
\vdots& \vdots &  \ddots & \vdots\\
a_1b_1^{k-1} & a_2b_2^{k-1} &\cdots & a_nb_n^{k-1}\\
\end{pmatrix}.
\]
The Euclidean dual code of $GRS_k(\mathbf{b},\mathbf{a})$ is also a GRS code.

For two integers $m$ and $n$, let $[m,n]$ denote the set $\{m,m+1,\ldots,n\}$, if $m\leq n$, and the empty set $\emptyset$, if $m>n$. An $[n,k,d]_q$ code $\cc$ is {\it cyclic} if $(c_{n-1},c_0,\ldots,c_{n-2})\in\cc$ for each $(c_0,c_1,\ldots,c_{n-1})\in\cc$. Each codeword $(c_0,c_1,\ldots,c_{n-1})$ of $\cc$ can be represented by a polynomial $c(x)=c_0+c_1x+\ldots+c_{n-1}x^{n-1} \in \Ff_q[x]$. This representation allows for the identification of $\cc$ as an ideal in $\Ff_q[x]/\langle x^n-1\rangle$. The monic polynomial $g(x)$ of degree $n-k$ in the ideal is a divisor of $x^n-1$ and is called the \emph{generator polynomial} of $\cc$.

If $r={\rm ord}_n(q)$ and $\alpha$ is a primitive $n^{\rm th}$ root of unity in $\Ff_{q^r}$, then the roots of $x^n-1$ are the elements $\alpha^j$ for $j\in[0,n-1]$. The collection $D=\{i\in[0,n-1]:g(\alpha^i)=0\}$ is the \emph{defining set} of a cyclic code $\cc$. Its complement $[0,n-1]\setminus D$ is the \emph{generating set} of $\cc$. Since $D$ is the union of some $q$-cyclotomic cosets modulo $n$, we can assume that the number of such $q$-cyclotomic cosets is $t$ and write $D=\bigcup_{\ell=1}^tC_{i_\ell}$ with $C_{i_\ell}$ being $q$-cyclotomic coset containing $i_\ell$. Using the $t\times n$ matrix
\[
H=\begin{pmatrix}
1 & \alpha^{i_1} &\cdots & \alpha^{i_1(n-1)}\\
1 & \alpha^{i_2} &\cdots & \alpha^{i_2(n-1)}\\
\vdots& \vdots &  \ddots & \vdots\\
1 & \alpha^{i_{t}} &\cdots & \alpha^{i_{t}(n-1)}\\
\end{pmatrix},
\]
a codeword $\mathbf{c}$ is in the cyclic code $\cc$ if and only if $H\mathbf{c}^{\top}=\mathbf{0}$. Let $\Tr^m_1$ denote the \emph{trace function} from $\Ff_{q^m}$ onto $\Ff_q$. The cyclic code $\cc$ has the following representation in terms of the trace function.
\begin{lem}{\rm (\cite{Delsarte1975})}\label{trace}
Let $\cc$ be a cyclic code of length $n$ over $\Ff_q$. Let $r={\rm ord}_n(q)$ and let $\alpha$ be a primitive $n^{\rm th}$ root of unity in $\Ff_{q^r}$. Let the generating set of $\cc$ be $\bigcup_{\ell=1}^sC_{i_\ell}$, where $C_{i_\ell}$ is the $q$-cyclotomic cosets modulo $n$ that contains $i_\ell$. Let $|C_{i_\ell}|= m_\ell$ for any $\ell\in[s]$. Then
\[
\cc=\left\{\left(\sum_{\ell=1}^s \Tr^{m_\ell}_1(\theta_\ell\alpha^{-ui_\ell})\right)_{u\in[0,n-1]}: \theta_\ell\in\Ff_{q^{m_\ell}},1\leq\ell\leq s\right\}.
\]
\end{lem}

The Hartmann-Tzeng bound in \cite{Hartmann1972} is a \emph{lower bound} on the minimum distance of a cyclic code.

\begin{lem}{\rm (Hartmann-Tzeng bound)}\label{HTbound} Let $\cc$ be an $[n,k,d]_q$ cyclic code with defining set $D$. Let $a$, $b$, and $c$ be integers such that $\gcd(b,n)=1$ and $\gcd(c,n)=1$. If $\{a+bi_1+ci_2 : i_1 \in[0,x-2] \mbox{ and } i_2\in[0,y]\}$ is contained in $D$, then $d\geq x+y$.
\end{lem}

The \emph{extended code} of the cyclic code $\cc$ is 
\[
E(\cc)=\left\{(c_0,\ldots,c_{n-1},c_n):(c_0,\ldots,c_{n-1})\in\cc \mbox{ and } \sum_{i=0}^n c_i =0\right\}.
\]
From \cite[Subsection 1.5.2]{huffman2003}, the code $E(\cc)$ has length $n+1$, dimension $k$, and parity-check matrix
\begin{equation}\label{extendedcyclic}
E(H)=\begin{pmatrix}
1 & 1 &\cdots & 1 & 1\\
1 & \alpha^{i_1} &\cdots & \alpha^{i_1(n-1)}& 0\\
\vdots& \vdots &  \ddots & \vdots&\vdots\\
1 & \alpha^{i_{t}} &\cdots & \alpha^{i_{t}(n-1)}& 0\\
\end{pmatrix}.
\end{equation}

We recall a special linear code that can serve as a powerful ingredient in a construction of Hermitian self-orthogonal codes. The code was introduced by Rains in \cite{Rains1999} and is defined by
\begin{equation}\label{eq:RainP}
P(\cc)=\left\{(a_1,\cdots,a_n)\in\Ff_q^n:\sum_{i=1}^na_iu_iv_i^q=0,\mbox{ for all}\ \mathbf{u},\mathbf{v}\in\cc\right\},
\end{equation}
where $\cc$ is an $[n,k,d]_{q^2}$ code, $\mathbf{u}=(u_1,\cdots,u_n)$, and $\mathbf{v}=(v_1,\cdots,v_n)$. Using $P(\cc)$, we can determine all Hermitian self-orthogonal codes that are monomially equivalent to $\cc$ or to the punctured codes of $\cc$.
\begin{lem}{\rm \cite{Ball2022}}\label{lem31}
Let $\cc$ be an $[n,k,d]_{q^2}$ code. There exists a punctured code of $\cc$ of length $m\leq n$ which is monomially equivalent to a Hermitian self-orthogonal code if and only if there is a codeword of weight $m$ in $P(\cc)$.
\end{lem}

Lemma \ref{lem31} motivates us to give a similar characterization of a GRS code whose Hermitian hull contains a GRS code. Let $\alpha$ be a primitive element of $\Ff_{q^2}$. Using $\mathbf{1}=(1,\ldots,1)$ and $\mathbf{b}=(\alpha^0,\ldots,\alpha^{q^2-2},0)$, we construct $GRS_k(\mathbf{b},\mathbf{1})$, with parameters $[q^2,k,q^2-k+1]_{q^2}$, whose Hermitian hull contains $GRS_\ell(\mathbf{b},\mathbf{1})$ as a subcode with $\ell\leq k$. We define the code $P(GRS_{k,\ell})$ as
\begin{equation}\label{GRS1}
P(GRS_{k,\ell})=\left\{(a_1,\cdots,a_n)\in\Ff_{q}^n:\sum_{i=1}^{q^2}a_iu_iv_i^q=0,\mbox{ for all } \mathbf{u}\in GRS_\ell(\mathbf{b},\mathbf{1}), \mbox{ for all } \mathbf{v}\in GRS_k(\mathbf{b},\mathbf{1})\right\}.
\end{equation}
By definition, $P(GRS_{k,\ell})$ has parity-check matrix
\[
H=\begin{pmatrix}
1 & 1 &\cdots & 1 & 1\\
1 & \alpha^{s_1} &\cdots & \alpha^{s_1(q^2-2)}& 0\\
\vdots& \vdots &  \ddots & \vdots&\vdots\\
1 & \alpha^{s_{t}} &\cdots & \alpha^{s_{t}(q^2-2)}& 0\\
\end{pmatrix},
\]
where $\{s_1,\ldots,s_t\}=\{i+qj:i\in[0,\ell-1],j\in[0,k-1]\}\setminus\{0\}$. Based on (\ref{extendedcyclic}), we know that $P(GRS_{k,\ell})$ is an extended code $E(\mathcal{D}_{k,\ell})$ of a $q$-ary cyclic code $\mathcal{D}_{k,\ell}$ of length $q^2-1$ with defining set 
\begin{equation}\label{eq:defset}
\{i+qj: i\in[0,\ell-1],j\in[\ell,k-1]\} \bigcup \{i+qj:i\in[\ell,k-1],j\in[0,\ell-1]\} \bigcup \left(\{i+qj:i\in[0,\ell-1],j\in[0,\ell-1]\} \setminus\{0\}\right).
\end{equation}
Note that the above defining set is $\{i+qj:i\in[0,\ell-1],j\in[0,\ell-1]\} \setminus\{0\}$ if $\ell=k$ as $[k,k-1]=\emptyset$. By Lemma \ref{HTbound}, the parameters of $E(\mathcal{D}_{k,\ell})$ are $[q^2,q^2-2\ell k+\ell^2,d\geq k+\ell-1]_q$.

\begin{lem}\label{lem32}
Let $\alpha$ be a primitive element of $\Ff_{q^2}$ and let $\mathbf{b}=(\alpha^0,\ldots,\alpha^{q^2-2},0)$. Let $P(GRS_{k,\ell})$ in (\ref{GRS1}) be the extended code $E(\mathcal{D}_{k,\ell})$ of $\mathcal{D}_{k,\ell}$ whose defining set is in (\ref{eq:defset}) and let $\ell\leq k$. Let $\widehat{\mathbf{a}}$ and $\widehat{\mathbf{b}}$ be vectors in $\Ff_{q^2}^m$ with $m\leq q^2$. There exists a $q^2$-ary GRS code $GRS_k(\widehat{\mathbf{b}},\widehat{\mathbf{a}})$ whose Hermitian hull contains its subcode $GRS_\ell(\widehat{\mathbf{b}},\widehat{\mathbf{a}})$ if and only if the following conditions are met. 
\begin{itemize}
    \item There exists a codeword $(x_1,\ldots,x_{q^2})$ of weight $m$ with nonzero entries $x_{i_1},\ldots,x_{i_m}$ in $E(\mathcal{D}_{k,\ell})$.
    \item The vector $\widehat{\mathbf{a}}$ is of the form $(a_{i_1},\ldots,a_{i_m})$ with $a_{i_j}^{q+1}=x_{i_j}$ for each $j\in[m]$.
    \item The vector $\widehat{\mathbf{b}}$ is obtained by deleting the $q^2-m$ entries whose coordinates are indexed by $[q^2]\setminus\{i_1,\ldots,i_m\}$ in $\mathbf{b}$.
\end{itemize}
\end{lem}
\begin{proof}
Let there be a codeword $(x_1,\ldots,x_n)$ of weight $m$ with nonzero entries $x_{i_1},\ldots,x_{i_m}$ in $E(\mathcal{D}_{k,\ell})$. There exists a nonzero element $a_{i_j} \in\Ff_{q^2}^*$ such that $x_{i_j}=a_{i_j}^{q+1}$ for each $j\in[m]$, since $x_{i_j}\in\Ff_q^*$. We note that $m\geq k+\ell-1\geq k$. Assume that $\mathbf{1}=(1,\ldots,1)$ is a vector of length $q^2$. Puncturing $GRS_k(\mathbf{b},\mathbf{1})$ and $GRS_\ell(\mathbf{b},\mathbf{1})$ on $S=[q^2]\setminus\{i_1,\ldots,i_m\}$, respectively, we obtain $GRS_k(\widehat{\mathbf{b}},\widehat{\mathbf{1}})$ and $GRS_\ell(\widehat{\mathbf{b}},\widehat{\mathbf{1}})$, both of length $m$. By (\ref{GRS1}), for any $(u_{i_1},\ldots,u_{i_m})\in GRS_\ell(\widehat{\mathbf{b}},\widehat{\mathbf{1}})$ and any $(v_{i_1},\ldots,v_{i_m})\in GRS_k(\widehat{\mathbf{b}},\widehat{\mathbf{1}})$, we have 
\[
\sum_{j=1}^mx_{i_j}u_{i_j}v_{i_j}^q = \sum_{j=1}^ma_{i_j}u_{i_j}(a_{i_j}v_{i_j})^q=0.
\]
Writing $\widehat{\mathbf{a}}=(a_{i_1},\ldots,a_{i_m})$, we arrive at $GRS_\ell(\widehat{\mathbf{b}},\widehat{\mathbf{a}}) \subseteq GRS_k(\widehat{\mathbf{b}},\widehat{\mathbf{a}})^{\perp_{\rm H}}$, which implies that $GRS_\ell(\widehat{\mathbf{b}},\widehat{\mathbf{a}})$ is a subcode of the Hermitian hull of $GRS_k(\widehat{\mathbf{b}},\widehat{\mathbf{a}})$. Conversely, the desired result follows by reversing the direction of the argument above.
\end{proof}

Lemma \ref{lem32} directly leads to the following technique to generate GRS codes whose Hermitian hulls are MDS.

\begin{thm}\label{GRSMDSHULL}
Let the notation be as in Lemma \ref{lem32}. If there is a codeword $(x_1,\ldots,x_{q^2})$ of weight $m$ with nonzero entries $x_{i_1},\ldots,x_{i_m}$ in $E(\mathcal{D}_{k,k-1})\setminus E(\mathcal{D}_{k,k})$, then the Hermitian hull of $GRS_k(\widehat{\mathbf{b}},\widehat{\mathbf{a}})$ is $GRS_{k-1}(\widehat{\mathbf{b}},\widehat{\mathbf{a}})$.
\end{thm}

\begin{remark}\label{remark1}
The respective dimensions of $E(\mathcal{D}_{k,k-1})$ and $E(\mathcal{D}_{k,k})$ are $q^2-k^2+1$ and $q^2-k^2$. Hence, there exists a nonzero codeword in $E(\mathcal{D}_{k,k-1})\setminus E(\mathcal{D}_{k,k})$ whenever $k\leq q$. If $k<q$, then Lemma \ref{trace} allows us to express $E(\mathcal{D}_{k,k-1})\setminus E(\mathcal{D}_{k,k})$ as
\begin{equation}\label{eqtr}
\left\{\left(c_0,\cdots,c_{q^2-2},\sum_{r=0}^{q^2-2}c_r\right):\theta_{k-1,k-1}\in\Ff_{q}^*,\theta_{t,t}\in\Ff_{q},t\in[k,q-1], \mbox{ and } \theta_{i,j}\in\Ff_{q^2} \mbox{ for } (i,j) \in T \right\},
\end{equation}
where $T=\bigcup_{i=k}^{q-1}\left\{(i,j):j\in[0,k-1] 
\bigcup[i+1,q-1]\right\}$, with $[q,q-1]=\emptyset$, and
\[
c_r=\sum_{t=k-1}^{q-1}\theta_{t,t}\alpha^{-r t(q+1)}+\sum_{(i,j)\in T}\Tr^2_1\left(\theta_{i,j}\alpha^{-r(i+qj))}\right)
\]
for each $r\in[0,q^2-2]$. We can then check that $\sum_{r=0}^{q^2-2}c_r=-\theta_{q-1,q-1}$. By \eqref{eqtr}, if  there exists a polynomial
\[
f(x) = x^{(k-1)(q+1)}+ \sum_{t=k}^{q-1}\theta_{t,t}x^{t(q+1)}+ \sum_{i=k}^{q-1} \sum_{j=0}^{k-1}\left(\theta_{i,j}x^{i+qj}+ \theta_{i,j}^qx^{j+qi}\right)+\sum_{i=k}^{q-2} \sum_{j=i+1}^{q-1}\left(\theta_{i,j}x^{i+qj}+ \theta_{i,j}^qx^{j+qi}\right)
\]
with $q^2-m$ distinct roots over $\Ff_{q^2}$, then $E(\mathcal{D}_{k,k-1})\setminus E(\mathcal{D}_{k,k})$ contains a codeword of weight $m$.
\end{remark}

Setting $k=q$ in Theorem \ref{GRSMDSHULL} gives a construction of GRS codes whose Hermitian hulls are MDS of dimension $q-1$.

\begin{thm}\label{GRScon1}
We use a primitive element $\alpha$ of $\Ff_{q^2}$ to stipulate vectors $\mathbf{1}=(1,\ldots,1)$ and $\mathbf{b}=(\alpha^0,\ldots,\alpha^{q^2-2},0)$ of length $q^2$. The Hermitian hull of $GRS_q(\mathbf{b},\mathbf{1})$ is $GRS_{q-1}(\mathbf{b},\mathbf{1})$.
\end{thm}
\begin{proof}
The respective dimensions of $E(\mathcal{D}_{q,q-1})$ and $E(\mathcal{D}_{q,q})$ are $1$ and $0$. By definition, 
\[
E(\mathcal{D}_{q,q-1})=\{(a,\ldots,a):a\in\Ff_{q^2}\}.
\]
By Theorem \ref{GRSMDSHULL}, the Hermitian hull of $GRS_q(\mathbf{b},\mathbf{1})$ is $GRS_{q-1}(\mathbf{b},\mathbf{1})$.
\end{proof}

Relying on (\ref{remark1}), we turn the task of finding a nonzero codeword of $E(\mathcal{D}_{k,k-1})\setminus E(\mathcal{D}_{k,k})$ into counting the number of the roots of a polynomial. As the $\theta_{i,j}$ in (\ref{eqtr}) traverses the elements of $\Ff_{q^2}$, we obtain the following constructions.

\begin{thm}\label{GRScon2}
Let $k$ be an integer such that $1<k<q$. Let $\alpha$ be a primitive element of $\Ff_{q^2}$. Using vectors 
\[
\mathbf{a}=(\alpha^{-0(k-1)},\ldots,\alpha^{-(q^2-2)(k-1)}) \mbox{ and } \mathbf{b}=(\alpha^0,\ldots,\alpha^{q^2-2})
\]
of length $q^2-1$, the Hermitian hull of $GRS_k(\mathbf{b},\mathbf{a})$ is $GRS_{k-1}(\mathbf{b},\mathbf{a})$.
\end{thm}
\begin{proof}
By setting $\theta_{k-1,k-1}=1$ and $\theta_{t,t}=\theta_{i,j}=0$, for each $t\in[k,q-1]$ and each $(i,j)\in T$ in (\ref{eqtr}), there exists
\[
\left(\alpha^{-0(k-1)(q+1)},\ldots,\alpha^{-(q^2-2)(k-1)(q+1)},0 \right) \in E(\mathcal{D}_{k,k-1})\setminus E(\mathcal{D}_{k,k}).
\]
By Theorem \ref{GRSMDSHULL}, 
${\rm Hull}_{\rm H}(GRS_k(\mathbf{b},\mathbf{a})) = GRS_{k-1}(\mathbf{b},\mathbf{a})$.
\end{proof}

\begin{thm}\label{GRScon3}
Let $k$ be a positive integer such that $k <q$. Let $\alpha$ be a primitive element of $\Ff_{q^2}$. Let 
\[
s=\gcd(k-1,q-1) \mbox{ and } B=\left\{i+\frac{q-1}{s}j:i\in[\frac{q-1}{s}-1],j\in[0,(q+1)s-1]\right\}.
\]
If $\mathbf{a}=((a_\ell)_{\ell\in B},a)$ and $\mathbf{b}=((\alpha^\ell)_{\ell\in B},0)$ are vectors of length $q^2-s(q+1)$ with $a^{q+1}=-1$ and $a_\ell^{q+1}=\alpha^{-\ell(k-1)(q+1)}-1$ for each $\ell\in B$, then the Hermitian hull of $GRS_k(\mathbf{b},\mathbf{a})$ is $GRS_{k-1}(\mathbf{b},\mathbf{a})$.
\end{thm}
\begin{proof}
In (\ref{eqtr}), let $\theta_{k-1,k-1}=1$, $\theta_{q-1,q-1}=-1$, and $\theta_{t,t}=\theta_{i,j}=0$ for each $t\in[k,q-2]$ and each $(i,j) \in T$. Then the corresponding codeword in $E(\mathcal{D}_{k,k-1})\setminus E(\mathcal{D}_{k,k})$ is
\[
\mathbf{c} = \left(\alpha^{-0(k-1)(q+1)}-1,\ldots, \alpha^{-(q^2-2)(k-1)(q+1)}-1,-1 \right).
\]
Since $s=\gcd(k-1,q-1)$, the zero components of $\mathbf{c}$ are $\alpha^{-i\frac{q-1}{s}(k-1)(q+1)}-1$ with $i\in[0,s(q+1)-1]$. By Theorem \ref{GRSMDSHULL}, we have 
${\rm Hull}_{\rm H}(GRS_k(\mathbf{b},\mathbf{a})) = GRS_{k-1}(\mathbf{b},\mathbf{a})$.
\end{proof}

\begin{thm}\label{GRScon4}
Let $k$ and $m$ be positive integers such that $k<q$ and $k-1<m<q-1$. Let $\alpha$ be a primitive element of $\Ff_{q^2}$. Let 
\[
s =\gcd(m-k+1,q-1) \mbox{ and } B=\left\{i+\frac{q-1}{s}j:i\in[\frac{q-1}{s}-1],j\in[0,(q+1)s-1]\right\}.
\]
If $\mathbf{a}=(a_\ell)_{\ell\in B}$ and $\mathbf{b}=(\alpha^\ell)_{\ell\in B}$ are vectors of length $(q+1)(q-1-s)$, with $a_\ell^{q+1}=\alpha^{-\ell(k-1)(q+1)}-1$ for each $\ell\in B$, then the Hermitian hull of $GRS_k(\mathbf{b},\mathbf{a})$ is $GRS_{k-1}(\mathbf{b},\mathbf{a})$.
\end{thm}
\begin{proof}
In (\ref{eqtr}), we put $\theta_{k-1,k-1}=1$, $\theta_{m,m}=-1$, and $\theta_{t,t}=\theta_{i,j}=0$ for each $t\in([k,q-1]\setminus\{m\})$ and each $(i,j)\in T$. The same argument as the one in the proof of Theorem \ref{GRScon3} leads us to the desired conclusion.
\end{proof}

\subsection{Optimal EAQECCs}

The families of optimal EAQECCs in the next theorem are explicitly constructed based on Theorems \ref{EAQMDSP}, \ref{GRScon1}, \ref{GRScon2}, \ref{GRScon3}, and \ref{GRScon4}. We once again note that optimal EAQECCs are codes that reach equality in the bound in \eqref{eq:QMDS}.

\begin{thm}\label{thmMDS}
Let $q$ be a prime power and let $k$ be an integer with $1<k<q$. Let $i=\gcd(k-1,q-1)$ and let $m$ be an integer such that $k-1<m<q-1$. If $j=\gcd(m-k+1,q-1)$, then we obtain the sixteen families of optimal EAQECCs in Table \ref{table1}.
\begin{table*}[ht]
\caption{Parameters of optimal $[[n,\kappa,\delta;c]]_q$ EAQECCs from Theorem \ref{thmMDS}.}
\label{table1}
\renewcommand{\arraystretch}{1.2}
\centering
\begin{tabular}{ccclll}
\toprule
No. & $n$  & $\kappa$  & $\delta$ &  $c$ & Range of $s$ \\
\midrule
$1$ & $q^2-s$  & $1$ & $q^2-q+1$ & $q^2-2q+1+s$ & $s\in \{0,1,\ldots,q-1\}$  \\
$2$ & $q^2-s$  & $1+s$ & $q^2-q+1-s$ & $q^2-2q+1$& \\
$3$ & $q^2-s$  & $q^2-2q+1$ & $q+1$ & $1+s$ \\
$4$ & $q^2-s$  & $q^2-2q+1+s$ & $q+1-s$ & $1$ \\
\midrule
$5$ & $q^2-1-s$  & $1$ & $q^2-k$ & $q^2-2k+s$ & $s\in \{0,1,\ldots,k-1\}$\\
$6$ & $q^2-1-s$  & $1+s$ & $q^2-k-s$ & $q^2-2k$ &  \\
$7$ & $q^2-1-s$  & $q^2-2k$ & $k+1$ & $1+s$ \\
$8$ & $q^2-1-s$  & $q^2-2k+s$ & $k+1-s$ & $1$ \\
\midrule
$9$ & $q^2-i(q+1)-s$  & $1$ & $q^2-i(q+1)-k+1$ & $q^2-i(q+1)-2k+1+s$ & $s\in \{0,1,\ldots,k-1\}$ \\
$10$ & $q^2-i(q+1)-s$  & $1+s$ & $q^2-i(q+1)-k+1-s$ & $q^2-i(q+1)-2k+1$& \\
$11$ & $q^2-i(q+1)-s$  & $q^2-i(q+1)-2k+1$ & $k+1$ & $1+s$\\
$12$ & $q^2-i(q+1)-s$  & $q^2-i(q+1)-2k+1+s$ & $k+1-s$ & $1$\\
\midrule
$13$ & $q^2-j(q+1)-1-s$  & $1$ & $q^2-j(q+1)-k$ & $q^2-i(q+1)-2k+s$ & $s\in \{0,1,\ldots,k-1\}$ \\
$14$ & $q^2-j(q+1)-1-s$  & $1+s$ & $q^2-j(q+1)-k-s$ & $q^2-i(q+1)-2k$& \\
$15$ & $q^2-j(q+1)-1-s$  & $q^2-i(q+1)-2k$ & $k+1$ & $1+s$\\
$16$ & $q^2-j(q+1)-1-s$  & $q^2-i(q+1)-2k+s$ & $k+1-s$ & $1$\\
\bottomrule
\end{tabular}
\end{table*}
\end{thm}

\begin{remark}
We see in \eqref{EAQcon} that optimal EAQECCs constructed by Lemma \ref{prop:two} always come in pairs. The \emph{initial} EAQECCs in Table \ref{table1} are given as Entries $4$, $8$, $12$, and $16$. They all have $s=0$ and  $c=1$. In most cases, their lengths are greater than $q+1$. Table \ref{table3} lists known optimal $[[n,\kappa,\delta;1]]_q$ EAQECCs of lengths $n>q+1$. The parameters of the initial codes in Entries $4$ and $8$ were known in the literature. However, we obtain optimal EAQECCs with new parameters by applying the propagation rules in Theorem \ref{EAQMDSP} to those initial codes. The initial codes in Entries $12$ and $16$ have new parameters. Thus, all eight families of codes in Entries $9$ to $16$ are new.
\begin{table*}[ht]
\caption{Parameters of known optimal $[[n,\kappa,\delta;1]]_q$ EAQECCs with $n>q+1$.}
\label{table3}
\renewcommand{\arraystretch}{1.5}
\centering
\begin{tabular}{ccclll}
\toprule
No. & $n$  & $\kappa$  & $\delta$ & Constraints & Reference \\
\midrule
$1$ & $\frac{q^2+1}{5}$  & $\frac{q^2+1}{5}-2k+3$ & $k$ & $q=10m+3$, with even $m$, and $2\leq k\leq 8m+61$, with even $d$ & \cite{Lu2018}  \\

$2$ & $\frac{q^2+1}{5}$  & $\frac{q^2+1}{5}-2k+3$ & $k$ & $q=10m+7$ with even $m$ &   \\

\midrule
$3$ & $z(q-1)$  & $z(q-1)-2k+1$ & $k+1$ & $1\leq z\leq q+1$, $\gcd(n,q)=1$, and $1\leq k\leq q-1$ & \cite{Guenda2017}  \\

\midrule

$4$ & $q^2+1$  & $q^2-2k+4$ & $k$ & $2\leq k \leq 2q$ with even $k$, & \cite{Fan2016}  \\
$5$ & $q^2$  & $q^2-2k+3$ & $k$ & $q+1\leq k \leq 2q-1$, &  \\
$6$ & $q^2-1$  & $q^2-2k+2$ & $k$ & $2\leq k \leq 2q-2$, &  \\

\midrule

$7$ & $\frac{q^2+1}{10}$  & $\frac{q^2+1}{5}-2k+3$ & $k$ & $q=10m+3$ and $2\leq k\leq 6m+2$ with even $d$ & \cite{Lu2018a}  \\

$8$ & $\frac{q^2+1}{10}$  & $\frac{q^2+1}{5}-2k+3$ & $k$ & $q=10m+7$ and $2\leq k\leq 6m+4$ with even $d$ &   \\

$9$ & $\frac{q^2+1}{h}$  & $\frac{q^2+1}{h}-2k+3$ & $k$ & $h\in\{3,5,7\}$, with $h\mid(q+1)$, and $\frac{q+1}{h}+1\leq k\leq \frac{(q+1)(h+3)}{2h}-1$ &   \\

\midrule

$10$ & $tr^z+s$  & $tr+s-2k-1$ & $k+1$ & $q=p^m$, $r=p^e$, with $e\mid m$, $1\leq t\leq r$,  & \cite{Fang2020}  \\
 &   &  &  &  $1\leq z\leq\frac{2m}{e}-1$, $1\leq k\leq \left\lfloor\frac{n-1+q}{q+1}\right\rfloor$, and $s\in\{0,1\}$ & \\

$11$ & $tr+s$  & $tr+s-2k-1$ & $k+1$ & $1\leq t\leq \frac{q-1}{a}$, with $a=\frac{r}{\gcd(r,q+1)}$, $r\mid(q^2-1)$, and $1\leq k\leq \left\lfloor\frac{n+q}{q+1}\right\rfloor$ &   \\
\bottomrule
\end{tabular}
\end{table*}
\end{remark}

\subsection{Optimal Subsystem Codes}

Applying Lemma \ref{QSCH} and Theorem \ref{QSCPMDS} on the GRS codes constructed in Section \ref{sec:CQC}-A yields eight families of optimal MDS subsystem codes.

\begin{thm}\label{subMDS}
Let $q$ be a prime power and let $k$ be an integer with $1<k<q$. Let $i=\gcd(k-1,q-1)$ and let $m$ be an integer such that $k-1<m<q-1$. If $j=\gcd(m-k+1,q-1)$, then we obtain the eight families of optimal subsystem codes in Table \ref{table2}.
\begin{table*}[ht]
\caption{The Parameters $[[n,\kappa,r,\delta]]_q$ of optimal subsystem codes from Theorem \ref{subMDS}.}
\label{table2}
\renewcommand{\arraystretch}{1.3}
\centering
\begin{tabular}{ccclll}
\toprule
No. & $n$  & $\kappa$  & $r$ &  $\delta$ & Range of $s$ \\
\midrule
$1$ & $q^2-s$  & $q^2-2q+1$ & $1+s$ & $q-s$ & $s \in \{0,1,\ldots,q-1\}$ \\
$2$ & $q^2-s$  & $q^2-2q+1+s$ & $1$ & $q-s$ & \\
\hline
$3$ & $q^2-1-s$  & $q^2-2k$ & $1+s$ & $k-s$ & $s \in \{0,1,\ldots,k-1\}$  \\
$4$ & $q^2-1-s$  & $q^2-2k+s$ & $1$ & $k-s$ &   \\
\hline
$5$ & $q^2-i(q+1)-s$  & $q^2-i(q+1)-2k+1$ & $1+s$ & $k-s$ & $s \in \{0,1,\ldots,k-1\}$ \\
$6$ & $q^2-i(q+1)-s$  & $q^2-i(q+1)-2k+1+s$ & $1$ & $k-s$ &  \\
\hline
$7$ & $q^2-j(q+1)-1-s$  & $q^2-j(q+1)-2k$ & $1+s$ & $k-s$ & $s \in \{0,1,\ldots,k-1\}$  \\
$8$ & $q^2-j(q+1)-1-s$  & $q^2-j(q+1)-2k+s$ & $1$ & $k-s$ &  \\
\bottomrule
\end{tabular}
\end{table*}
\end{thm}
\begin{proof}
The listed parameters were obtained by applying Lemma \ref{QSCH}, Theorems \ref{GRScon1}, \ref{GRScon2}, \ref{GRScon3}, \ref{GRScon4}, and \ref{QSCPMDS}.
\end{proof}

\begin{remark}
In most cases, the lengths of our optimal subsystem codes are greater than $q+1$. We list the parameters of known optimal subsystem codes of lengths $n>q+1$ in Table \ref{table4}. Comparing the entries in Tables \ref{table2} and \ref{table4} confirms that the optimal subsystem codes in Table \ref{table2} have new parameters, except for the codes in Entries $2$ and $4$.

\begin{table*}[ht]
\caption{The Parameters of known $[[n,\kappa,r,\delta]]_q$ optimal subsystem codes}
\label{table4}
\renewcommand{\arraystretch}{1.3}
\centering
\begin{tabular}{cccllll}
\toprule
No. & $n$  & $\kappa$  & $r$ &  $\delta$ & Constraints & Reference \\
\midrule
$1$ & $q^2-1-s$ & $q^2-2k-1-\ell+s$ & $\ell$ & $k+1-s$ & $0\leq k<q-1$, $s \in \{0,1,\ldots,k\}$, and $0\leq \ell< q^2-2k-1$ & \cite{SALAHA.2009}\\
$2$ & $q^2-s$ & $q^2-2k-2-\ell+s$ & $\ell$ & $k+2-s$ & $0\leq k<q-1$, $s \in \{0,1,\ldots,k+1\}$, and $0\leq \ell < q^2-2k-2$ & \\
\hline
$3$ & $2^{2m}+1$  & $(2^m-1)^4$ & $4$ & $2^m-1$ & $m\geq 1$ & \cite{Qian2013} \\
\bottomrule
\end{tabular}
\end{table*}
\end{remark}

\section{Concluding remarks}\label{sec:conclu}

General propagation rules in the stabilizer framework for quantum error-control have been known since almost the beginning of studies into the topic. For qubit, most of the rules are discussed in \cite{Calderbank1998}. The rules extend naturally to cover qudit cases. Controlling quantum errors in the respective frameworks of entanglement-assisted and quantum subsystem requires thoughtful modifications. The corresponding classical coding structures that allow for the derivation of the parameters of the relevant quantum codes have not been as extensively studied as in the stabilizer framework. 

In this work, we have presented new propagation rules on quantum codes in the entanglement-assisted and subsystem frameworks. We examine classical codes that can serve as ingredients in the constructions of quantum codes with desirable properties. It turns out that focusing on the Hermitian hulls of punctured and shortened linear codes with optimal parameters leads us to new propagation rules. With obvious modifications, we can derive analogous results when switching the inner product from Hermitian to Euclidean. Specific constructions to demonstrate the efficacy of our approach have also been proposed. We have included tables that list the parameters of known optimal EAQECCs and quantum subsystem codes for ease of reference.

In Section \ref{sec:CQC}, we started with generalized Reed-Solomon (GRS) codes and established conditions that ensure their Hermitian hulls are maximum distance separable. Puncturing and shortening maintain the optimality properties in the resulting codes and in their corresponding quantum codes. For implementation purposes in actual quantum channels, however, one would prefer to keep to a qubit or a qutrit setup, that is, keeping to $q \in \{2,3\}$ as in Tables \ref{table1} to \ref{table4}. Under the Hermitian inner product, this means keeping the classical code ingredients to be linear over $\Ff_4$ or $\Ff_9$. By definition, this requirement severely limits the lengths of the quantum codes that we can construct. Three open directions naturally suggest themselves. 
\begin{enumerate}
    \item Determine if there exist linear codes, not necessarily MDS, whose hulls, under some valid inner product, are MDS of dimension $> 1$. If yes, propose valid constructions of such codes.
    \item Derive measures of optimality other than the Singleton or Singleton-like upper bounds. The bounds should ideally be better suited for the already implementable qubit setup and, to a lesser extent, for the qutrit setup.
    \item Construct new quantum codes which are demonstrably optimal or better than currently known. In particular, it will be especially interesting to construct optimal $q$-ary quantum codes with minimum quantum distance $\delta > q+1$. In the stabilizer QECCs, we learn from \cite{Ball2022} that the $q$-ary quantum codes derived from GRS codes have minimum quantum distance $\delta\leq q$.
\end{enumerate}

There are likely to be other propagation rules to be discovered even in the already familiar scenarios for quantum error-control. Formulating them explicitly and constructing good codes as a result constitutes a valuable contribution.

\section*{Acknowledgments}
The authors thank Markus Grassl for insightful discussions. His comments on earlier drafts of this work led to sharper analysis on whether or not a general propagation rule requires the initial quantum code to be pure.


\end{document}